\documentclass[5p]{elsarticle}

\usepackage{amsmath}
\usepackage{amssymb}
\usepackage[utf8]{inputenc}

\usepackage{enumerate}
\usepackage[all]{xy}
\usepackage[vlined,algoruled]{algorithm2e}

\journal{Automatica}

\biboptions{authoryear}

\def \R{{\mathbb R}}
\def \N{{\mathbb N}}

\def \KL{\mathcal{KL}}

\def \ki{\mathcal{K}_{\infty}}

\def \U{\mathcal{U}}
\def \S{\mathcal{S}}
\def \T{\mathcal{T}}

\def \A{\mathcal{A}}

\def\K{\mathcal{K}}
\def\Ki{\K_{\infty}}
 \def\mer{\hfill $\circ$}

 \def\AL{\A\mathcal{L}}

\DeclareMathOperator*{\esssup}{ess.sup}

\newtheorem{teo}{Theorem}[section]
\newtheorem{lema}[teo]{Lemma}

\newtheorem{claim}{Claim}
\newdefinition{defin}[teo]{Definition}
\newdefinition{as}{Assumption}
\newdefinition{rem}[teo]{Remark}
\newproof{proof}{\textbf{Proof}}

\renewcommand{\qed}{$\hfill\blacksquare$}
\renewcommand{\subsection}[1]{\stepcounter{subsection}\vspace{\baselineskip}\noindent\textbf{\thesubsection. #1}}

\begin{document}

\begin{frontmatter}

\title{Nonrobustness of asymptotic stability of impulsive systems with inputs}


\author[HHaddress]{Hernan Haimovich\corref{mycorrespondingauthor}}
\cortext[mycorrespondingauthor]{Corresponding author}
\ead{haimovich@cifasis-conicet.gov.ar}

\author[ITBA]{Jos\'e L. Mancilla-Aguilar}
\ead{jmancill@itba.edu.ar}

\address[HHaddress]{International Center for Information and Systems Science (CIFASIS), CONICET-UNR, Ocampo y Esmeralda, 2000 Rosario, Argentina.}
\address[ITBA]{Instituto Tecnol\'ogico de Buenos Aires, Av. E. Madero 399, Buenos Aires, Argentina.}

\begin{abstract}
  Suitable continuity and boundedness assumptions on the function $f$ defining the dynamics of a time-varying nonimpulsive system with inputs are known to make the system inherit stability properties from the zero-input system. Whether this type of robustness holds or not for impulsive systems was still an open question. By means of suitable (counter)examples, we show that such stability robustness with respect to the inclusion of inputs cannot hold in general, not even for impulsive systems with time-invariant flow and jump maps. In particular, we show that zero-input global uniform asymptotic stability (0-GUAS) does not imply converging input converging state (CICS), and that 0-GUAS and uniform bounded-energy input bounded state (UBEBS) do not imply integral input-to-state stability (iISS). We also comment on available existing results that, however, show that suitable constraints on the allowed impulse-time sequences indeed make some of these robustness properties possible.
\end{abstract}


\end{frontmatter}


\section{Introduction}
\label{sec:introduction}

Impulsive systems are dynamic systems whose state evolves continuously most of the time but may exhibit jumps (discontinuities) at isolated time instants \citep{lakbai_book89}. The set of time instants when jumps occur are part of the impulsive system definition. We consider impulsive systems where the continuous dynamics is governed by a differential equation, characterized by the flow map, 
 and where the state value immediately after a jump is given by a static equation, namely the jump map.

The stability properties of impulsive systems with or without inputs depend on the interplay between the continuous and the impulsive behaviors \citep{heslib_auto08} given by the flow map, the jump map, and the set of impulse times. These properties have been extensively studied and several sufficient conditions for asymptotic, input-to-state and integral input-to-state stability were obtained, even for systems with time-varying flow and jump maps and under the presence of time delays
\citep[see][among others]{heslib_auto08,%
chezhe_auto09,%
chezhe_auto09b,%
liuliu_auto11,%
briseu_scl12,%
wanwan_ietcta13,%
dasmir_siamjco13,%
liudou_scl14,%
gaowan_nahs16,%
barval_mjm16,%
dasfek_nahs17,%
manhai_tac19arxiv,%
honzha_ijrnc19,%
fekbaj_auto19}. Most of these references assume the existence of a Lyapunov-type function which may provide some degree of robustness with respect to the inclusion of disturbances or modeling errors. In addition, some of these also explicitly address robustness of stability. 

This paper is concerned with a fundamental question: whether asymptotic stability of an impulsive system without disturbances or with zero input may guarantee some kind of robustness of the system with respect to the inclusion of inputs/disturbances. 

For nonimpulsive (time-varying) systems under reasonable continuity assumptions on the function $f$ defining the dynamics (local Lipschitz continuity, uniformly with respect to the time variable), the uniform asymptotic stability of the  system when the input or the disturbance is identically zero (0-UAS) guarantees various kinds of robustness properties of the system with inputs/disturbances. For example, it is known that 0-UAS implies that the system with inputs/disturbances is totally stable (TS) \citep{hahn_book67} which roughly speaking means that trajectories corresponding to small initial conditions and small inputs or disturbances remain near the equilibrium point. Another robustness property implied by the 0-UAS property is the so-called converging input converging state (CICS): every bounded trajectory which lies in the domain of attraction and corresponds to an input or disturbance that approaches zero must also converge to the equilibrium \citep{sontag_tac03,ryason_scl06, manhai_scl17}.
%
Yet another of these properties is captured by the fact that the combination of global uniform asymptotic stability under zero input/disturbance (0-GUAS) and a uniformly bounded-energy input/bounded state (UBEBS) property implies integral input-to-state stability \citep[iISS,][]{sontag_scl98}, as proved initially by \\
 \citet{angson_dc00} for time-invariant systems and extended to time-varying systems in \citet{haiman_tac18}. A consequence of this fact can be loosely stated as follows: for some suitable way of measuring input energy, if the input energy is finite, then the state will converge to zero. This property can also be interpreted as providing some robustness of stability with respect to the inclusion of inputs, but provided that the state remains bounded under inputs of bounded energy. 

In this paper, we address impulsive systems where both the flow and jump maps could be time-varying and depend on external inputs. We show that even under stronger uniform boundedness, and state and input Lipschitz continuity assumptions on the flow and jump maps, 0-GUAS  implies neither CICS nor TS, and 0-GUAS and UBEBS do not imply iISS. We show that this is so even when 0-GUAS is uniform not only with respect to initial time but also over all possible impulse-time sequences and, moreover, also when the flow and jump maps are time-invariant and the former is input-independent. A very salient feature of our negative results is that they do not depend on how the input energy is measured; in other words, they are valid for any UBEBS and iISS gains, even of course when these could be different from each other. The results that we provide thus clearly illustrate that the stated nonrobustness of impulsive systems is of a very profound nature. This lack of robustness is directly related neither to how the input may enter into the system equations nor to the regularity of the flow and jump maps; it is indeed related to the fact that the definition of 0-GUAS usually considered in the literature of impulsive systems is too weak for guaranteeing any meaningful robustness property.

For (time-invariant) hybrid systems, it is known that 0-GUAS is robust with respect to the inclusion of inputs or disturbances \citep{goesan_book12}. For this to happen, however, stability must take hybrid time into account, thus causing decay towards the equilibrium set not only when continuous time elapses but also whenever jumps occur. In this regard, we have already shown that if, in addition to elapsed time, the number of jumps is taken into account in the definition of asymptotic stability, then 0-GUAS and UBEBS imply iISS  \citep[see][]{haiman_rpic19,haiman_auto19b_arxiv}. The main contribution of this paper is thus to show that taking the number of jumps into account within the stability definition is unavoidable for the stated robustness to be possible.

\textbf{Notation.} $\N$, $\N_0$, $\R$, and $\R_{\ge 0}$ denote the natural numbers, the nonnegative integers, the reals, and the nonnegative reals, respectively. $|x|$ denotes the Euclidean norm of $x \in \R^p$. We write $\alpha\in\K$ if $\alpha:\R_{\ge 0} \to \R_{\ge 0}$ is continuous, strictly increasing and $\alpha(0)=0$, and $\alpha\in\Ki$ if, in addition, $\alpha$ is unbounded. We write $\beta\in\KL$ if $\beta:\R_{\ge 0}\times \R_{\ge 0}\to \R_{\ge 0}$, $\beta(\cdot,t)\in\Ki$ for any $t\ge 0$ and, for any fixed $r\ge 0$, $\beta(r,t)$ monotonically decreases to zero as $t\to \infty$. $\lceil r \rceil$ is the least integer greater than or equal to $r\in \R$.

%

\section{Problem Statement}
\label{sec:prob-stat}
Consider the time-varying impulsive system with inputs $\Sigma$ defined by the equations
\begin{subequations}
	\label{eq:is}
	\begin{align}
	\label{eq:is-ct}
	\dot{x}(t) &=f(t,x(t),u(t)),\phantom{h^-}\quad\text{for } t\notin \gamma,    \displaybreak[0] \\
	\label{eq:is-st}
	x(t) &=h(t,x(t^-),u(t)),\phantom{f} \quad\text{for } t\in \gamma,
	\end{align}
\end{subequations}
where $t\ge 0$, $x(t)\in \R^n$, $u(t)\in \R^m$, $f$ and $h$ are functions from $\R_{\ge 0}\times \R^n\times \R^m$ to $\R^n$ such that $f(t,0,0)=0$ and $h(t,0,0)=0$ for all $t\ge 0$, and $\gamma=\{\tau_k\}_{k=1}^{\infty}$, with $0<\tau_1<\tau_2<\cdots$ and $\lim_{k\to \infty}\tau_k=\infty$, is the impulse-time sequence.
By ``input'', we mean a Lebesgue measurable and locally bounded function $u:[0,\infty)\to \R^m$; we denote by $\U$ the set of all the inputs. An input $u$ could represent, e.g., a control input or a disturbance input. We define for convenience $\tau_0=0$.

A solution of $\Sigma$ corresponding to an initial time $t_0\ge 0$, an initial state $x_0\in \R^n$ and an input $u\in \U$ is a right-continuous 
function $x:[t_0,T_x)\to \R^n$ such that $x(t_0)=x_0$ and: 
\begin{enumerate}[i)]
	\item $x$ is locally absolutely continuous on each nonempty interval $J$ of the form $J=[\tau_k,\tau_{k+1})\cap [t_0,T_x)$, with $k\ge 0$, and $\dot{x}(t)=f(t,x(t),u(t))$ for almost all $t\in J$; and \label{item:solflow}
	\item for all $\tau_k\in (t_0,T_x)$, the left limit $x(\tau_k^-)$ exists and is finite, and $x(\tau_k) = h(\tau_k,x(\tau_k^-),u(\tau_k))$.\label{item:soljump}
\end{enumerate}
The solution $x$ is said to be maximally defined if no other solution $y:[t_0,T_y)\to \R^n$ satisfies $y(t) = x(t)$ for all $t\in [t_0,T_x)$ and has $T_y > T_x$. A solution $x$ is forward complete if $T_x=\infty$, and $\Sigma$ is forward complete if every maximal solution of $\Sigma$ is forward complete.

We will use $\T(t_0,x_0,u)$ to denote the set of maximally defined solutions of $\Sigma$ corresponding to initial time $t_0$, initial state $x_0$ and input $u$.

An important problem in control theory is understanding the dependence of state trajectories on the inputs, in particular when the inputs are bounded or when they converge to zero as $t\to \infty$. In order to make the latter precise, given an input $u \in \U$, an interval $I\subset \R_{\ge 0}$, and functions $\rho_1,\rho_2\in\Ki$, we define
\begin{align} 
\| u_I \|_{\infty} &:= \max\left\{ \esssup_{t\in I} |u(t)| , \sup_{t\in \gamma\cap I} |u(t)| \right\}, \label{eq:iss-norm}\\
\| u_I \|_{\rho_1,\rho_2} &:= \int_I \rho_1(|u(s)|) ds + \sum_{s\in \gamma\cap I} \rho_2(|u(s)|). \label{eq:iiss-norm}
\end{align}
When $I=[t_0,\infty)$ we simply write $u_{t_0}$ instead of $u_I$. In both input bounds the values of $u$ at the instants $t\in\gamma$ are explicitly taken into account, since these values may instantaneously affect the state trajectory. 

The following stability properties give characterizations of the behavior of the state trajectories when the inputs are bounded, converge to zero, or are identically zero. In what follows, $\mathbf{0}$ denotes the identically zero input.
\begin{defin} 
  The impulsive system $\Sigma$ is 
  \begin{enumerate}[a)] 
  \item zero-input globally uniformly asymptotically stable (0-GUAS) if there exists $\beta \in \KL$ such that for all $t_0\ge 0$, $x_0\in \R^n$, and $x\in \T(t_0,x_0,\mathbf{0})$, it happens that $x$ is forward complete and for all $t\ge t_0$
    \begin{align}
      \label{eq:0-guas}
      |x(t)| &\le \beta \left (|x_0|,t-t_0 \right);
    \end{align}
  \item uniformly bounded-energy input/bounded state (UBEBS) 
    if $\Sigma$ is forward complete and there exist $\alpha,\rho_1,\rho_2\in\ki$ and $c\ge 0$ such that 
    \begin{align}
      \label{eq:cubebs}
      \alpha(|x(t)|) &\le |x_0| + \| u_{(t_0,t]} \|_{\rho_1,\rho_2} + c
    \end{align}
    for all $t\ge t_0\ge 0$, $x_0\in \R^n$, $u\in \U$, and $x\in \T(t_0,x_0,u)$; 
  \item integral input-to-state stable (iISS) if $\Sigma$ is forward complete and there exist $\beta \in \KL$ and $\alpha,\rho_1,\rho_2 \in \ki$ such that for all $t_0\ge 0$, $x_0\in \R^n$, $u\in \U$ and $x\in \T(t_0,x_0,u)$, it happens that for all $t\ge t_0$,
    \begin{align} 
      \label{eq:ciiss}
      \hspace{-3mm}\alpha(|x(t)|) \le \beta \left (|x_0|,t-t_0 \right) + \| u_{(t_0,t]} \|_{\rho_1,\rho_2};
    \end{align}
  \item converging-input converging-state (CICS) if every forward complete and bounded solution $x\in \T(t_0,x_0,u)$, with $t_0\ge 0$, $x_0\in \R^n$, and $u \in \U$ such that $\|u_t\|_{\infty}\to 0$ as $t\to \infty$, satisfies $x(t)\to 0$ as $t\to \infty$; \label{item:cics}
  \item totally stable (TS) if $f(t,\xi,\mu)\equiv f_0(t,\xi)+\mu$, $h(t,\xi,\mu)\equiv h_0(t,\xi)+\mu$ and, for every $\varepsilon>0$ there exists $\delta>0$ such that every solution $x\in \T(t_0,x_0,u)$, with $t_0\ge 0$, $x_0\in \R^n$ with $|x_0|<\delta$, and $u \in \U$ such that $\|u_{t_0}\|_{\infty}<\delta$, satisfies $|x(t)|<\varepsilon$ for all $t\in [t_0,T_x)$. \label{item:ts}
  \end{enumerate}
\end{defin}
\begin{rem} 
  The definition of total stability given above is a natural generalization, to impulsive systems, of the one usually considered in the literature of ordinary differential equations \citep[see][Chapter VII]{hahn_book67}. \mer
\end{rem}
\begin{rem}
  \label{rem:uniform}
  We say that (\ref{eq:is}) is 0-GUAS (or UBEBS) uniformly with respect to a set $\S$ of impulse-time sequences if every system $\Sigma$ defined by (\ref{eq:is}) with $\gamma\in\S$ is 0-GUAS (or UBEBS) and, moreover, the bound (\ref{eq:0-guas}) [or (\ref{eq:cubebs})] holds with the same $\beta$ (or $\rho_1,\rho_2,\alpha$ and $c$) for every such system.\mer
\end{rem}

From the definition of these stability properties, it easily follows that iISS implies $0$-GUAS and UBEBS. For nonimpulsive time-varying systems and under appropriate assumptions on the flow map $f$, it was proved that $0$-GUAS and UBEBS imply iISS \citep[Theorem 1]{haiman_tac18}, and that $0$-GUAS implies TS \citep[Theorem 56.4]{hahn_book67} and CICS \citep[Section 3.2]{manhai_scl17}. The question that naturally arises is thus whether the same implications remain true for impulsive systems. 

In \citet[Theorem 3.2]{haiman_aadeca18} it was shown that $0$-GUAS and UBEBS imply iISS for time-varying impulsive systems, assuming that the impulse-time sequence satisfies the so-called uniform incremental boundedness (UIB) condition \citep[Definition 3.2]{haiman_aadeca18}, and in \citet{haiman_rpic19} that the same implication holds without the UIB condition if the $0$-GUAS property is strengthened.

The main result of the current paper is to show that the mentioned implications do not remain valid if $0$-GUAS is understood in the usual sense and the UIB condition is not assumed. We will do so through counterexamples in the next section.

\section{Main results}
\label{sec:ce}

In this section, we show that even if the flow and jump maps are time-invariant, $0$-GUAS implies neither CICS nor TS, and $0$-GUAS and UBEBS do not imply iISS. In addition, we will show that these negative results remain true when 0-GUAS and UBEBS are uniform over all possible impulse-time sequences (as per Remark~\ref{rem:uniform}) but the jump map is allowed to be time-varying.

\subsection{Impulsive system equations}
\label{sec:impuls-sys-eqs}

Consider the scalar impulsive system $\Sigma$, with a single input, of the form (\ref{eq:is}) with
\begin{align}
  f(t,\xi,\mu) &= -\xi, \label{eq:fdef}\\
  \label{eq:hdef}
  h(t,\xi,\mu) &= \hat h(t,\xi)+\mu,\\
  \label{eq:hath}
   \hat h(t,\xi) &= \begin{cases}
      \bar h(|\xi|) &\text{if }|\xi|\le e^{-\sigma(t)},\\
      |\xi|&\text{otherwise},
    \end{cases}
\end{align}
where $\bar h:\R_{\ge 0}\to \R_{\ge 0}$ is defined by
\begin{align}
 \label{eq:barh}
\bar h(r) =
\begin{cases}
0 &\text{if }r=0, \\
e^{\lceil\ln r\rceil}  &\text{if }{\scriptstyle\lceil\ln r\rceil-0.5 < \ln r\le 0},\\
{\scriptstyle(1+e^{0.5})r - e^{\lceil\ln r \rceil-0.5}}
&\text{if } {\scriptstyle\ln r \le  \lceil\ln r \rceil -0.5 \le 0},\\
r &\text{otherwise},
\end{cases}
\end{align}
and the function $\sigma$ is defined as follows.
For $i\in\N$, let $S_i$ be the increasing and finite sequence containing the first $i$ natural numbers, i.e. $S_1=\{1\}$, $S_2=\{1,2\}$, etc., and construct the infinite sequence $\{\sigma_k\}_{k=1}^\infty$ by concatenating $S_1,S_2,\ldots$, so that the first elements of $\{\sigma_k\}$ are $1,1,2,1,2,3,1,2,3,4,\ldots$ Then, define
\begin{align*}
  \sigma(t):= \sigma_i-1 \quad\text{when }t\in [i-1,i), \ i\in\N.
\end{align*}
The function $\sigma$ can be equivalently defined as follows:
\begin{align}
  \label{eq:sigmaeqdef}
  \sigma(t) &= i,\quad \text{for }t \in [s_j+i,s_j+i+1), &i &=0,\ldots,j,\\
 \label{eq:sj}
 s_j &:= \sum_{i=0}^j i, &j &\in\N_0.
\end{align}

Figure~\ref{fig:h0xi0} illustrates the function $h(0,\cdot,0)$.
\begin{figure}[!h]
  \centering
  \includegraphics[width=\columnwidth]{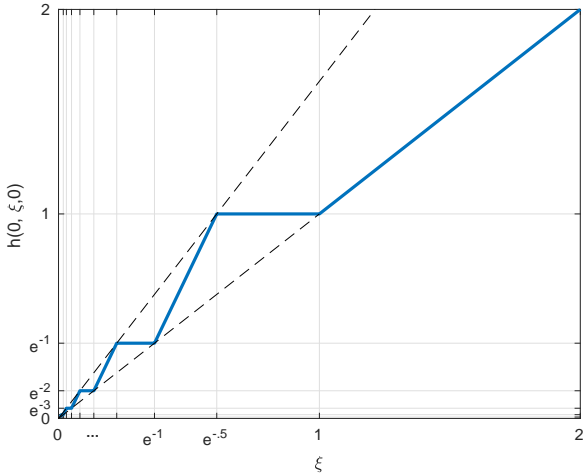}
  \caption{The function $h(0,\xi,0)$ for $\xi \ge 0$.}
  \label{fig:h0xi0}
\end{figure}
The function $h$ in (\ref{eq:hdef}) is nondecreasing in the absolute value of the state variable when the other variables are fixed, i.e. $h$ satisfies $h(t,|\xi_1|,\mu) \ge h(t,|\xi_2|,\mu)$ for all $t\ge 0$, $\mu \in \R$ and $|\xi_1| \ge |\xi_2| \ge 0$.
Note that system $\Sigma$ with $f$ and $h$ defined, respectively, by (\ref{eq:fdef}) and (\ref{eq:hdef}), and $\gamma$ any impulse-time sequence, is forward complete, since the differential equation (\ref{eq:is-ct}) has no finite escape time. 
\begin{rem}
  \label{rem:time-invariant}
  If the impulse-time sequence $\gamma$ is such that $\sigma(t) = 0$ for every $t\in\gamma$, then the evolution of (\ref{eq:is}) with (\ref{eq:fdef})--(\ref{eq:hdef}) becomes equivalent to that arising when $h$ in (\ref{eq:hdef}) is replaced by the time-invariant jump map $h_{\text{ti}}(\xi,\mu) = h(0,\xi,\mu)$. Therefore, in such a case the system (\ref{eq:is}) with (\ref{eq:fdef})--(\ref{eq:hdef}) is equivalent to an impulsive system having time-invariant flow and jump maps.\mer
\end{rem}
The following property of $\sigma$ will be instrumental in establishing $0$-GUAS.
\begin{lema}
	\label{lem:phi-eps-delta}
	For every $k\in \N_0$ there exists $\bar T_k>0$ such that for every $t_0 \ge 0$ there exists $t_0 \le t^*\le t_0+\bar T_k$ such that $\sigma(s) \ge k+1 $ for all $t^* \le s \le t^* + 1$.
\end{lema}
\begin{proof}
	Let $k\in \N_0$. Set $\bar T_k = k+1 + s_{k+2}$, with $s_j=\sum_{i=0}^j i$. According to (\ref{eq:sigmaeqdef}), we have $\sigma(s) \ge k+1$ for $s\in [s_{j} + k+1, s_{j+1})$ for all $j\ge k+1$.
	
	If $t_0 \le \bar T_k$, take $t^*=\bar T_k$. 
	
	If $t_0 > \bar T_k$, let $\kappa := \max \{j \in \N_0 : s_j \le t_0 \}$. Clearly, $\kappa \ge k+2$ and $s_{\kappa}\le t_0<s_{\kappa+1}$. If $s_{\kappa}\le t_0 \le s_{\kappa}+ k+1$, take $t^* = s_{\kappa} + k+1$.
	
	If $s_{\kappa} + k+1<t_0 \le s_{\kappa+1} - 1$, take $t^*=t_0$. 
	
	Otherwise, take $t^*=s_{\kappa+1} + k+1$.
	\qed
\end{proof}

\subsection{$0$-GUAS and UBEBS}
\label{sec:0-guas-ubebs}

\begin{lema} 
  \label{lem:0guas-ubebs}
  The impulsive system $\Sigma$ in (\ref{eq:is}), with $f$ and $h$ defined by (\ref{eq:fdef}) and (\ref{eq:hdef}), respectively, and $\gamma$ any impulse-time sequence, is 0-GUAS and UBEBS. Moreover, (\ref{eq:is}) with (\ref{eq:fdef})--(\ref{eq:hdef}) is 0-GUAS and UBEBS both uniformly with respect to the set of all impulse-time sequences.
\end{lema}
\begin{proof}
We will prove that $\Sigma$ is $0$-GUAS by establishing that: i) $\Sigma$ is $0$-input globally uniformly stable ($0$-GUS), i.e. there exists a class-$\Ki$ function $\nu$ such that for all $x\in \T(t_0,x_0,\mathbf{0})$ with $t_0\ge 0$ and $x_0\in \R$ it happens that $|x(t)|\le \nu(|x_0|)$ for all $t\ge t_0$; and ii) for all $0<\varepsilon< r$ there exists $T=T(r,\varepsilon)\ge 0$ such that for all $x\in \T(t_0,x_0,\mathbf{0})$ with $t_0\ge 0$ and $|x_0|\le r$ we have that $|x(t)|\le \varepsilon$ for some $t\in [t_0,t_0+T]$.

Since for all $x\in \T(t_0,x_0,\mathbf{0})$, $|x|\in \T(t_0,|x_0|,\mathbf{0})$, and $x(t)\equiv 0$ when $x_0=0$, we only have to establish i) and ii) for positive initial conditions $x_0$. Let $x\in \T(t_0,x_0,\mathbf{0})$ with $t_0\ge 0$ and $x_0>0$. Then $x(t)> 0$ for all $t\ge t_0$. If $x_0 > 1$ then $x(t)\le x_0$ for all $t\ge t_0$. Suppose that the latter is not true. Since $x(t)$ is nonincreasing between consecutive impulse times, the first time $t$ for which $x(t)>x_0$ must be an impulse time $t\in \gamma\cap (t_0,\infty)$. For such a time $t$ we have that $x(t^-)\le x_0$ and $x(t)=h(t,x(t^{-}), 0)$. Since $h$ is non decreasing in the absolute value of its second argument, we have that $x(t)=h(t,x(t^{-}), 0)\le h(t,x_0,0)=x_0$, which is absurd. Suppose now that $0<x_0 \le 1$. Let $k(x_0)=\lceil \ln x_0 \rceil \le 0$. Then $x(t)\le e^{k(x_0)}$ for all $t\ge t_0$. Suppose on the contrary that $x(t)> e^{k(x_0)}$ for some $t\ge t_0$. Since $x$ is nonincreasing between consecutive impulse-times and $x(t_0)\le e^{k(x_0)}$, the first time $t\ge t_0$ for which $x(t)> e^{k(x_0)}$ has to be an impulse-time $t\in \gamma\cap (t_0,\infty)$. Then we have that $x(s)\le  e^{k(x_0)}$ for all $s\in [t_0,t)$ and $x(t)=h(t,x(t^-),0)>e^{k(x_0)}$. Since $x(t^-)\le  e^{k(x_0)}$ and $h$ is nondecreasing in the absolute value of its second argument, we have that $h(t,x(t^-),0)\le h(t,e^{k(x_0)},0)=e^{k(x_0)}$, arriving to a contradiction.

Since the function $\bar \nu(r)=r$ for $r> 1$ and $\bar \nu(r)=e^{\lceil \ln r \rceil}$ for $0<r\le 1$ is non decreasing and $\lim_{r\to 0^+} \bar \nu(r)=0$, there exists $\nu \in \Ki$ such that $\bar \nu (r)\le \nu(r)$ for all $r\ge 0$. In consequence $\Sigma$ satisfies item i) with such a function $\nu$.

For establishing ii) we first prove the following.
\begin{claim} \label{clm:A} Let $k\in \N_0$. Then there exists $T_k>0$ such that for all $x\in \T(t_0,x_0,\mathbf{0})$ with $t_0\ge 0$ and $0<x_0\le e^{-k}$ there is a $t\in [t_0,t_0+T_k]$ such that $x(t)\le e^{-(k+1)}$.
 \end{claim}
 \textbf{Proof of Claim~\ref{clm:A}:} Suppose that $x\in \T(t_0,x_0,\mathbf{0})$ with $t_0\ge 0$ and $0<x_0\le e^{-k}$. Note that $0<x(t)\le e^{k(x_0)}\le e^{-k}$ for all $t\ge t_0$, with $k(x_0)$ as defined above. Let $\bar T_k>0$ and $t^*\in [t_0,t_0+\bar T_k]$ be the quantities coming from Lemma \ref{lem:phi-eps-delta}. Suppose that $x(t)> e^{-(k+1)}$ for all $t\in [t^*,t^*+1]$. Since $\sigma(t)\ge k+1$ on $[t^*,t^*+1]$, from the definitions of $f$ and $h$ we have that $x(t)=x(t^*)e^{-(t-t^*)}$ for all $t\in [t^*,t^*+1]$. Therefore $x(t^*+1)=x(t^*)e^{-1}\le  e^{-(k+1)}$, which is absurd. In consequence there exists $t\in [t^*,t^*+1]$ so that $x(t)\le e^{-(k+1)}$, and the claim follows by taking $T_k=\bar T_k+1$. \mer

We proceed to prove ii). Let $0<\varepsilon<r$. Suppose that $r\le 1$. Let $t_0\ge 0$, $0<x_0\le r$ and $x\in \T(t_0,x_0,\mathbf{0})$.  Pick $k_1,k_2\in \N_0$ such that $e^{-k_2}<\varepsilon <r \le e^{-k_1}$. Let $T_{k_1}, T_{k_1+1},\ldots,T_{k_2-1}$ be the quantities coming from Claim \ref{clm:A} corresponding to $k=k_1,\ldots,k_2-1$. Applying Claim \ref{clm:A} in a recursive way, it follows that there exists $t_0 \le t_1 \le \cdots \le t_{k_2-k_1}$, with $t_1-t_0\le T_{k_1},\ldots, t_{k_2-k_1}-t_{k_2-k_1-1}\le T_{k_2-1}$ such that $x(t_j)\le e^{-(k_1+j)}$ for all $j=1,\ldots,k_2-k_1$. In consequence $x(t_{k_2-k_1})\le e^{-k_2}<\varepsilon$ and $t_{k_2-k_1}-t_0\le \sum_{i=0}^{k_2-k_1-1}T_{k_1+j}$. So item ii) holds in this case with $T(r,\varepsilon)=\sum_{i=0}^{k_2-k_1-1}T_{k_1+j}$. 

Suppose now that $r>1$. If $x_0\le 1$, then $x(t)<\varepsilon$ for some $t\in [t_0,t_0+T(1,\varepsilon)]$. If $x_0>1$, then by solving the equations (\ref{eq:is}) on the interval $I=[t_0,t_0+\ln x_0]$  it follows that $x(t)=x_0 e^{-(t-t_0)}$ for all $t\in I$. So $x(t_0+\ln x_0)=1$. In consequence, there exists $t\in [t_0+\ln x_0,t_0+\ln x_0+T(1,\varepsilon)]$ such that $x(t)<\varepsilon$, and ii) follows with $T(r,\varepsilon)=\ln r+T(1,\varepsilon)$ in this case.

The existence of a function $\beta \in \KL$ as in the definition of $0$-GUAS follows from i) and ii) and the steps used in the proof of \cite[Lemma 2.5]{linson_jco96}. The fact that the same $\beta$ can be used for every impulse-time sequence $\gamma$ follows from the fact that neither the function $\nu$ in i) nor the time $T(r,\varepsilon)$ in ii) depends on the specific $\gamma$. 

The system $\Sigma$ is UBEBS because for all $x\in \T(t_0,x_0,u)$ with $t_0\ge 0$, $x_0\in \R$ and $u\in \U$, 
\begin{align} \label{eq:ubebsce}
  |x(t)| &\le |x_0| + \|u_{(t_0,t]}\| + 1
\end{align}
holds for $\|u_I\|$ defined as
\begin{align*}
  \|u_I\| &= \int_I|u(t)| dt + \sum_{t\in\gamma\cap I} |u(t)|.
\end{align*}
For a contradiction, suppose there is $t\ge 0$ such that (\ref{eq:ubebsce}) does not hold. Since $|x|$ is nonincreasing between consecutive impulse times, the first time $t^*$ for which (\ref{eq:ubebsce}) is not true must satisfy $t^*\in \gamma \cap (t_0,\infty)$. Then $|x(t)|\le |x_0| + \|u_{(t_0,t]}\| + 1$ for all $t\in [t_0,t^*)$ and  $|x(t^*)|> |x_0| + \|u_{(t_0,t^*]}\| + 1$. Since  $|x({t^*}^-)|\le |x_0| + \|u_{(t_0,t^*)}\| + 1$, it follows that 
\begin{align*}
|x(t^*)|&\le h(t^*,|x({t^*}^-)|,u(t^*))\\
&\le h(t^*,|x_0| + \|u_{(t_0,t^*)}\| + 1,u(t^*))\\
&\le |x_0| + \|u_{(t_0,t^*)}\|+1+ |u(t^*)| \\
&= |x_0| + \|u_{(t_0,t^*]}\| + 1,
\end{align*}
Which is absurd. Here we have used that $h$ is nondecreasing in its second argument and the fact that $h(t,\xi,\mu)=|\xi|+\mu$ if $|\xi|\ge 1$. Since~(\ref{eq:ubebsce}) holds for every impulse-time sequence $\gamma$, then we have also established UBEBS uniformly with respect to the set of all impulse-time sequences. \qed
\end{proof}
\begin{rem}
	It can be verified that the functions $f$ and $g(t,\xi,\mu)=h(t,\xi,\mu)-\xi$ belong to $\AL$ with the class $\AL$ as defined in \citet[Definition 3.1]{haiman_aadeca18}. So, $\Sigma$ is iISS for every impulse time sequence $\gamma$ which is UIB according to \citet[Theorem 3.2]{haiman_aadeca18}. See also \citet{haiman_rpic19}. 
\end{rem}

\subsection{$0$-GUAS implies neither CICS nor TS}
\label{sec:0-guas-cics}

In this section, we show that for the impulse-time sequence $\gamma^*$ defined below, the system $\Sigma$ with $f$ and $h$ given by (\ref{eq:fdef})--(\ref{eq:hdef}) is not CICS and that if we consider the function $f(t,\xi,\mu)=-\xi+\mu$ instead of $f(t,\xi)=-\xi$, then the resulting system $\Sigma$ is not TS.  

For $N\in \N$, consider the finite sequence $\S_N=\{\tau_{N,k}\}_{k=0}^{2^N}$ with $\tau_{N,k}=s_N+\frac{k}{2^{N+1}}$, where $s_N$ is defined by (\ref{eq:sj}). Define $\gamma^*$ as the sequence obtained by concatenating $\mathcal{S}_1,\mathcal{S}_2,\ldots$, so that the first elements of $\gamma^*$ are $s_1,s_1+\frac{1}{4},s_1+\frac{2}{4}, s_2,s_2+\frac{1}{8},s_2+\frac{2}{8}, s_2+\frac{3}{8}, s_2+\frac{4}{8},  \ldots$ The sequence $\gamma^*$ has the property that $\sigma(t) = 0$ for every $t\in\gamma^*$ [recall (\ref{eq:sigmaeqdef})--(\ref{eq:sj})]. According to Remark~\ref{rem:time-invariant}, then the system $\Sigma$ in (\ref{eq:is}) with (\ref{eq:fdef})--(\ref{eq:hdef}) and $\gamma = \gamma^*$ becomes equivalent to a system with time-invariant flow and jump maps.
\begin{teo}
  \label{thm:ncics} 
  The system $\Sigma$ with $f$ and $h$ given by (\ref{eq:fdef})--(\ref{eq:hdef}), and with $\gamma=\gamma^*$, is not CICS.
\end{teo}
\begin{proof}
  Consider the input $u^*$ defined as follows:
  \begin{align}
    u^*(t) &=\begin{cases} \mu_N &\text{if } t=\tau_{N,k}, \\
      0 &\text{otherwise},
    \end{cases}\\
    \mu_N &=\frac{1-e^{-\Delta_N}}{1-e^{-1/2}},\quad \Delta_N=2^{-(N+1)}. 
  \end{align}
  Note that $\|u^*_{s_N}\|_{\infty}=\mu_N$ for all $N$, since $|u^*(t)|\le \mu_N$ for all $t> s_N$ and $u^*(\tau_{N,k})=\mu_N$ for all $k=1,\ldots, 2^N$. Therefore $\|u^*_t\|_{\infty}\to 0$ as $t\to \infty$ because $\|u^*_t\|$ is nonincreasing in $t$ and $\mu_N\to 0$.

Let $x$ be the unique solution of $\Sigma$ corresponding to $u^*$ and satisfying the initial condition $x(0)=0$. We claim that $x(s_N+1/2)\ge 1$ for all $N\in \N$.

By solving the equations of $\Sigma$ on $[0,s_1]$ it is clear that $x(s_1)=\mu_1>0$. So, $x(t)>0$ for all $t\ge s_1 = 1$. Let $N\in \N$, then, if $I_N=[s_N,s_N+1/2]$, we have that $\gamma^*\cap I_N=\mathcal{S}_N$ and that $\sigma(t)=0$ for all $t\in I_N$. Let $\xi_{N,k}=x(\tau_{N,k})$ for $k=0,\ldots, 2^N$. Then, for all $0\le k\le 2^{N}-1$, $\xi_{N,k+1}=\bar h(\xi_{N,k}e^{-\Delta_N})+\mu_N$. Since $\bar h(r)\ge r$ for all $r\ge 0$, we have that for all $1\le k \le 2^N-1$
\begin{align*}
\xi_{N,k+1}\ge  \xi_{N,k}e^{-\Delta_N}+\mu_N.
\end{align*}
By induction on $k$ it can be proved that
\begin{align*}
\xi_{N,k}\ge \mu_N \sum_{j=0}^{k-1} e^{-j\Delta_N}.
\end{align*}
Therefore 
\begin{align*}
x(s_N+1/2)=\xi_{N,2^N}\ge \mu_N \sum_{j=0}^{2^N-1} e^{-j\Delta_N}=1.
\end{align*}
Since $s_N\to \infty$, it follows that $x(t)$ does not converge to $0$ as $t\to \infty$, and thus $\Sigma$ is not CICS. \qed
\end{proof}
\begin{teo}\label{thm:nts} The system $\Sigma$ with $f(t,\xi,\mu)=-\xi+\mu$, $h$ given by (\ref{eq:hdef}) and $\gamma=\gamma^*$, is $0$-GUAS but not TS.
\end{teo}
\begin{proof} It is clear that $\Sigma$ is $0$-GUAS, since its zero-input system is the same as that of the system considered in Lemma \ref{lem:0guas-ubebs}. Consider the input $u^*$ defined in the proof of Theorem \ref{thm:ncics}. Given $\delta>0$, let $N$ be so that $\mu_N<\delta$. Then $\|u^*_{s_N}\|_{\infty}<\delta$. Let $x\in \T(s_N,0,u^*)$. By proceeding as in the proof of Theorem \ref{thm:ncics} it follows that $|x(s_N+1/2)|\ge 1$, showing that $\Sigma$ is not TS.\qed
\end{proof}

\subsection{$0$-GUAS and UBEBS do not imply iISS}

Next, we prove that $\Sigma$ with $f$ and $h$ defined by (\ref{eq:fdef}) and (\ref{eq:hdef}), respectively, and with $\gamma=\gamma^*$ is not iISS.

\begin{teo}\label{thm:niiss}
	Let $\Sigma$ be the system considered in Theorem \ref{thm:ncics}. Then $\Sigma$ is not iISS.
\end{teo}
Theorem \ref{thm:niiss} is a straightforward consequence of the following result.
\begin{lema}\label{lem:nubebs0}
	Let $\Sigma$ be the system considered in Theorem \ref{thm:ncics}.
Let $\rho_1,\rho_2\in\Ki$ and write $\|u\|=\|u\|_{\rho_1,\rho_2}$. Let $\delta_1,\delta_2>0$. Then, there exist $t_0,t,x_0,u$ such that $0\le t_0 \le t$, $|x_0| \le \delta_1$, $\|u\|\le\delta_2$, and the system solution $x$ corresponding to initial time $t_0$, initial condition $x_0$ and input $u$ satisfies
  \begin{align*}
    |x(t)| &\ge e^{-1}.
  \end{align*}
\end{lema}
\begin{proof}
  If $\delta_1 \ge e^{-1}$, then the result follows trivially with $t = t_0$. So, consider $\delta_1 < e^{-1}$. Define 
  \begin{align*}
    n_0 &= -\lfloor\ln\delta_1\rfloor,\\ 
    \bar\mu &= \min\left\{ \rho_2^{-1}(\delta_2/n_0), e^{-n_0+1}-e^{-n_0} \right\},\\
    \Delta &= \min\left\{ \ln\dfrac{1}{1-\bar\mu}, \ln\dfrac{1+\sqrt{e}\bar\mu}{1+\bar\mu} \right\}.
  \end{align*}
  Note that $0 < \Delta < 1/2$.
  Consider the input construction algorithm given as Algorithm~\ref{alg:main}.
\begin{algorithm}[!h]
  \label{alg:main}
  \SetTitleSty{texrm}{\normalsize}
  \SetKwComment{tcc}{\% }{}
  \KwData{$\delta_1,\Delta,\bar\mu$}
  \KwOut{$F$, $\{\xi_k\}_{k=1}^F$, $\{\mu_k\}_{k=1}^F$, $\{k_i\}_{i=1}^{n_1}$}
  \SetKwBlock{Init}{begin}{end}
  \Init(Initialization){
    $\xi_0 = e^{\lfloor\ln\delta_1\rfloor}$,
    $k \leftarrow 0$, $i\leftarrow 0$;\hfill(I)\\
  }
  \Repeat{$\xi_k \ge e^{-1}$}
         {$\ell_k = -\lceil\ln\xi_k - \Delta \rceil$;\hfill(R1)\\
           $k \leftarrow k+1$;\hfill(R2)\\
         \eIf{$-\ell_{k-1} - 0.5 \le \ln \xi_{k-1} - \Delta \le -\ell_{k-1}$} 
           {$\mu_k = \bar\mu$;\hfill(Ri1)\\
             $i\leftarrow i+1$;\hfill(Ri2)\\
           $k_i = k$;\hfill(Ri3)\\}
           {$\mu_k = 0$;\hfill(Re)\\}
           $\xi_{k} =\bar h(\xi_{k-1} e^{-\Delta})+\mu_k$;\hfill(R3)\\
         }
  \caption{Input sequence construction}
\end{algorithm}
The rationale for this algorithm is as follows. It will later be shown that this algorithm generates a sequence of state values (namely $\{\xi_k\}$) that will constitute a lower bound for the state evolution at the impulse times. This input construction algorithm generates a zero input (namely $\mu_k = 0$) whenever the unforced system dynamics pushes the state farther from the origin. A small input ($\mu_k = \bar\mu$) is generated only when necessary to make the subsequent unforced dynamics keep steering the state farther.

The algorithm begins by setting an initial condition for the state lower bound sequence $\xi_0 = e^{\lfloor \ln\delta_1 \rfloor} \le \delta_1$, at the initialization step (I). Then, in the repeat block at (R1), it happens that $\ell_0 = -\lceil\ln\xi_0 -  \Delta \rceil = -\lceil\lfloor\ln\delta_1\rfloor -  \Delta \rceil = -\lfloor\ln\delta_1\rfloor = n_0$. At (R2), $k$ is set to $k=1$. Then, the if condition initially holds, because  $-\ell_0 = \lceil\ln\xi_0 -  \Delta \rceil \ge \ln \xi_0 - \Delta = \lfloor\ln\delta_1\rfloor - \Delta = -\ell_0 - \Delta > -\ell_0 - 0.5$. Consequently, at the first iteration, corresponding to $k=1$, (Ri1) to (Ri3) will be executed so that $\mu_1 = \bar\mu$, $i$ is set to $i=1$ and $k_1 = k = 1$. At (R3), we have $\xi_1 = \bar h(\xi_0 e^{-\Delta}) + \mu_1 = \bar h(e^{\lfloor \ln \delta_1 \rfloor - \Delta}) + \bar\mu = e^{\lfloor \ln\delta_1 \rfloor} + \bar\mu$, where we have used (\ref{eq:barh}). Recalling the definition of $\bar\mu$, it follows that $\xi_0 < \xi_1 \le e^{-n_0 + 1}$. 

We claim that this algorithm finishes in a finite number of steps $F$ that depends on $\delta_1$ and $\delta_2$, and that $n_1$, the number of iterations at which $\mu_k\neq 0$, satisfies $n_1 \le n_0$ so that $\sum_{k=1}^F \rho_2(\mu_k)\le n_1 \rho_2(\bar\mu) \le \delta_2$. Whenever $-\ell_{k-1} -0.5 \le \ln \xi_{k-1} - \Delta \le -\ell_{k-1}$ (this holds for $k=1$), then according to (Ri1) and (R3) in Algorithm~\ref{alg:main}, and (\ref{eq:barh}), then 
\begin{align}
  \xi_{k} &= e^{-\ell_{k-1}} + \bar\mu \ge \xi_{k-1} e^{-\Delta} + \bar\mu \ge \xi_{k-1} (1-\bar\mu) + \bar\mu\notag\\
  \label{eq:xikdif1}
  &= \xi_{k-1} + (1-\xi_{k-1}) \bar\mu \ge \xi_{k-1} + (1-e^{-1})\bar\mu > \xi_{k-1}
\end{align}
provided $\xi_{k-1} \le e^{-1}$ (otherwise, the algorithm would have stopped). Hence, $\ell_k \le \ell_{k-1}$. Also in this case, we have 
\begin{align}
  \label{eq:xikD1}
  \xi_{k}e^{-\Delta} &= (e^{-\ell_{k-1}} + \bar\mu) e^{-\Delta} \ge (e^{-\ell_{k-1}} + \bar\mu) \frac{1+\bar\mu}{1+\sqrt{e}\bar\mu}. 
\end{align}
The function $\phi(r) = \frac{1+r}{1+\sqrt{e}r}$ is strictly decreasing in $\R_{\ge 0}$ and therefore $\phi(\bar\mu) > \phi(a\bar\mu)$ for every $a > 1$. Take $a = e^{\ell_{k-1}-0.5}$, which satisfies $a > 1$ because $\ell_{k-1} \ge 1$ whenever $\xi_{k-1} \le e^{-1}$ (otherwise the algorithm would have stopped), and operate on (\ref{eq:xikD1}) to obtain
\begin{align}
  \label{eq:xikeD}
  \xi_{k}e^{-\Delta} &> (e^{-\ell_{k-1}} + \bar\mu) \frac{1+e^{\ell_{k-1}-0.5} \bar\mu}{1+ e^{\ell_{k-1}}\bar\mu} = e^{-\ell_{k-1}} + e^{-0.5}\bar\mu.
\end{align}
It follows that
\begin{align}
  \label{eq:mellk}
  -\ell_k = \lceil \ln\xi_k - \Delta \rceil = \lceil \ln(\xi_k e^{-\Delta}) \rceil > -\ell_{k-1}.
\end{align}
In addition, by definition of $\bar\mu$ and provided $\ell_{k-1} \le n_0$ then
\begin{align*}
  \xi_{k}e^{-\Delta} &< \xi_k = e^{-\ell_{k-1}} + \bar\mu \le e^{-\ell_{k-1}} + e^{-n_0+1} - e^{-n_0}\\
  &\le e^{-\ell_{k-1}} + e^{-\ell_{k-1}+1} - e^{-\ell_{k-1}} = e^{-\ell_{k-1}+1}.
\end{align*}
Application of $\ln$ to the latter inequality yields
\begin{align*}
  \ln \xi_k - \Delta \le -\ell_{k-1} + 1,
\end{align*}
and since the right-hand side is integer valued, then also
\begin{align}
  \label{eq:mellk2}
  -\ell_k = \lceil \ln \xi_k - \Delta \rceil \le -\ell_{k-1} + 1.
\end{align}
From~(\ref{eq:mellk}) and (\ref{eq:mellk2}), we reach
\begin{align}
  \label{eq:ellk}
  \ell_k = \ell_{k-1} - 1.
\end{align}
We have thus shown that if (Ri1) to (Ri3) are executed in Algorithm~\ref{alg:main}, so that $\mu_k =\bar\mu$, then the value $\xi_k$ subsequently set at (R3) must satisfy
\begin{align*}
  \lceil \ln\xi_k - \Delta \rceil = 1 - \ell_{k-1}
\end{align*}
and (\ref{eq:ellk}) will hold at (R1). 
As a consequence, the first iteration number $k$ at which it happens that $\ln \xi_{k-1} - \Delta < -\ell_{k-1}-0.5$ ($k\ge 2$ because this does not happen at $k=1$), then it must be true that $\xi_{k-1} e^{-\Delta} > e^{-\ell_{k-2}} + e^{-0.5}\bar\mu = e^{-\ell_{k-1}-1} + e^{-0.5}\bar\mu$, as follows from (\ref{eq:xikeD}) and (\ref{eq:ellk}). In this case the if condition in Algorithm~\ref{alg:main} is not satisfied, (Re) is executed, and at (R3) it will happen that
\begin{align}
  \xi_{k} &= (1+\sqrt{e}) \xi_{k-1} e^{-\Delta} - e^{\lceil \ln\xi_{k-1} - \Delta \rceil - 0.5}\notag\\
  &=  (1+e^{0.5}) \xi_{k-1} e^{-\Delta} - e^{0.5}e^{-\ell_{k-1}-1}\notag\\
  &\ge (1+e^{0.5}) \xi_{k-1} e^{-\Delta} + e^{0.5}(- \xi_{k-1} e^{-\Delta} + e^{-0.5}\bar\mu)\notag\\
  \label{eq:xikdif2}
  &= \xi_{k-1} e^{-\Delta} + \bar\mu \ge \xi_{k-1} + (1-e^{-1})\bar\mu > \xi_{k-1}. 
\end{align}
Although in this case $\xi_k > \xi_{k-1}$, from (Re), (R3) and the definition of $\bar h$, then $\xi_k = \bar h(\xi_{k-1}e^{-\Delta}) \le e^{-\ell_{k-1}}$ and hence still $\ell_k = \ell_{k-1}$. As a consequence, $\xi_k e^{-\Delta} > \xi_{k-1}e^{-\Delta} > e^{-\ell_{k-1}-1} + e^{0.5} \bar\mu = e^{-\ell_{k}-1} + e^{0.5} \bar\mu$. 
Therefore, the inequality $\xi_{k-1} e^{-\Delta} > e^{-\ell_{k-1}-1} + e^{-0.5}\bar\mu$ holds whenever $\ln \xi_{k-1} - \Delta < -\ell_{k-1}-0.5$ and the above derivations show that $\xi_k > \xi_{k-1}$ whenever $\ln \xi_{k-1} - \Delta < -\ell_{k-1}-0.5$.
The sequence $\{\xi_k\}$ generated by Algorithm~\ref{alg:main} is thus strictly increasing and $\xi_{k+1} - \xi_k \ge (1-e^{-1})\bar\mu > 0$, as follows from (\ref{eq:xikdif1}) and~(\ref{eq:xikdif2}). We can thus bound the maximum number of iterations required as $$F \le \left\lceil \dfrac{e^{-1} - e^{-n_0}}{(1-e^{-1})\bar\mu} \right\rceil.$$

Since the sequence $\{\xi_k\}_{k=0}^F$ is strictly increasing, then the integer sequence $\{\ell_k\}_{k=0}^F$ is nonincreasing. Consider the sequence $\{k_i\}_{i=1}^{n_1}$. We have that $\mu_k\neq 0$ if and only if $k=k_i$ for some $i$, and, from the first part of the proof, that $k_1=1$, $\ell_{k_1}=n_0-1$ and $\ell_{k_i}=\ell_{k_i-1}-1$ for all $i=1,\ldots, n_1$. Since $\xi_{k_{n_1}-1}<e^{-1}$, we have that 
\begin{align*}
-\ell_{k_{n_1}-1}-0.5+\Delta \le \ln \xi_{k_{n_1}-1}<-1, 
\end{align*} 
and then $\ell_{k_{n_1}-1}\ge 1$. In consequence  $\ell_{k_{n_1}}=\ell_{k_{n_1}-1}-1 \ge 0$. Since $n_0-1\ge \ell_{k_1}-\ell_{k_{n_1}}\ge n_1-1$, it follows that $n_1\le n_0$.
	

Next, consider the quantities produced by Algorithm~\ref{alg:main}. Let $N \in \N$ be such that $\Delta_N=2^{-(N+1)}<\Delta$ and $2^{N}>F$. Define the input $u$ via $u(\tau_{N,k})=\mu_k$ for $k=1,\ldots,F$ and $u(t)=0$ otherwise. Note that $\|u\| \le \delta_2$. Consider the solution $x$ corresponding to initial time $s_N$, initial condition $\delta_1$ and input $u$. We have $x(\tau_{N,1}^-) = \delta_1 e^{-\Delta_N} \ge \delta_1 e^{-\Delta}$. 
Then 
      \begin{align*}
      x(\tau_{N,1}) &=h(\tau_{N,1},\delta_1 e^{-\Delta_N},u(\tau_{N,1}))\\
      &\ge h(\tau_{N,1},\delta_1 e^{-\Delta},u(\tau_{N,1})) = \bar h(\delta_1 e^{-\Delta})+\mu_1
      \end{align*}
        where the last equality follows from the fact that $\sigma(\tau_{N,k})=0$ for all $\tau_{N,k}$ since $\tau_{N,k}\in [s_N,s_N+1/2]$. 
      
      Then $x(\tau_{N,1}) \ge  \bar h(\xi_0e^{-\Delta})+\mu_1 = \xi_1$.
      By induction, we can prove that $x(\tau_{N,i}) \ge \xi_i$ for $i=1,2,\ldots,F$.
      Suppose that $x(\tau_{N,i}) \ge \xi_i$ for some $i$. This already holds for $i=1$. Then, $x(\tau_{N,i+1}^-) = x(\tau_{N,i}) e^{-\Delta_N} \ge x(\tau_{N,i}) e^{-\Delta}$ and 
      \begin{align*}
        x(\tau_{N,i+1}) &= h(\tau_{N,i+1},x(\tau_{N,i+1}^-),u(\tau_{N,i+1}))\\
        &\ge h(\tau_{N,i+1},\xi_ie^{-\Delta},u(\tau_{N,i+1}))\\       
        &= \bar h(\xi_i e^{-\Delta})+\mu_{i+1} = \xi_{i+1},
      \end{align*}
      where we have used the properties and definition of $h$ and the facts that $\xi_i\le 1$ and $\sigma(\tau_{N,i+1})=0$.
      
      As a consequence, it will happen that $x(\tau_{N,F}) \ge \xi_F \ge e^{-1}$, and the result is established with $\xi=\delta_1$, $t_0 = s_N$ and $t=\tau_{N,F}$.
  \qed
\end{proof}
{\bf Proof of Theorem \ref{thm:niiss}.} Suppose that $\Sigma$ is iISS. Then there exist $\beta\in \KL$ and $\alpha$, $\rho_1$ and $\rho_2\in \Ki$ such that (\ref{eq:ciiss}) holds. Pick any $\delta>0$ so that $\beta(\delta,0)+\delta<\alpha(e^{-1})$. Then, for all $t_0\ge 0$, $|x_0|\le \delta$, $u\in \U$ such that $\|u\|_{\rho_1,\rho_2}\le \delta$ and $t\ge t_0$, if $x\in \T(t_0,x_0,u)$ then for all $t\ge t_0$
\begin{align*}
\alpha(|x(t)|)&\le \beta(|x(t_0)|,t-t_0)+\|u\|_{\rho_1,\rho_2}\\
&\le \beta(\delta,0)+\delta <\alpha(e^{-1}).
\end{align*}
Therefore $|x(t)|< e^{-1}$ for all $t\ge t_0$. Since the latter contradicts Lemma \ref{lem:nubebs0}, it follows that $\Sigma$ is not iISS. \qed

We emphasize that since $\gamma = \gamma^*$ is the impulse-time sequence considered in Theorems~\ref{thm:ncics},~\ref{thm:nts} and~\ref{thm:niiss}, and since $\sigma(t) = 0$ for every $t\in\gamma^*$, then we have actually shown that the given negative results hold for an impulsive system with time-invariant flow and jump maps (recall Remark~\ref{rem:time-invariant}). 

\section{Discussion}
\label{sec:discussion}

The results obtained in Section \ref{sec:ce} show that the $0$-GUAS property as usually defined for impulsive systems is too weak for the system with inputs/disturbances to inherit any kind of meaningful stability with respect to small inputs. In fact, Theorems~\ref{thm:ncics} and~\ref{thm:nts} show that the state of an impulsive system may be not small even when the initial condition and the inputs are small in magnitude. Lemma~\ref{lem:nubebs0} shows that irrespective of the way in which the energy of an input is defined, the magnitude of a solution corresponding to an arbitrarily small initial condition and to an input with arbitrarily small energy may be not necessarily small. 

One main reason for this lack of robustness is the fact that even if Zeno behavior is ruled out from admissible impulse-time sequences, i.e. impulse-time sequences cannot have finite limit points, impulses can occur arbitrarily frequently as time progresses. This is the case for the sequence $\gamma^*$ defined in Section~\ref{sec:0-guas-cics}. When the initial time is not fixed, such as in the currently considered time-varying case, an appropriately large initial time can make impulses as frequent as desired even if the elapsed time $t-t_0$ is small. Although in practice it would be reasonable to assume that impulses cannot occur infinitely often, in some settings one cannot upper bound the number of impulses a priori in relation to elapsed time. Therefore, the $0$-GUAS property as usually defined for impulsive systems is, mathematically, not very useful in the analysis/design of real world systems where the existence of modeling errors or disturbances is unavoidable, unless the number of impulses could be bounded in relation to elapsed time. Note that this type of bound on the number of impulses exists in the case of fixed dwell-time or average dwell-time sequences and, most generally, UIB sequences as per \citet{haiman_aadeca18}.


These facts show the need for a stronger stability concept for impulsive systems. One way of strengthening stability is to mimic that considered in the theory of hybrid systems \citep{goesan_book12} by taking into account, in the decaying term appearing in (\ref{eq:0-guas}), the number of impulse-time instants contained in the interval $(t_0,t]$, namely $N(t_0,t)$. This is achieved, for example, replacing the term $\beta(|x(t_0)|,t-t_0)$ by $\beta(|x(t_0)|,t-t_0+N(t_0,t))$ \citep[see][]{manhai_tac19arxiv}. It can be easily shown that 0-GUAS in the usual sense implies this stronger 0-GUAS whenever the number of impulses that occur can be bounded in relation to elapsed time, a property that we referred to as the uniform incremental boundedness (UIB) of the impulse-time sequences \citep[see][]{haiman_aadeca18,manhai_tac19arxiv}.

In a forthcoming paper we will show that with this stronger definition of stability, $0$-GUAS implies the CICS, the TS and the BEICS properties \citep{jayrya_tac10,jayawa_siamjco10}, where the latter is defined as follows: The system is bounded-energy-input converging-state (BEICS) if there exist $\rho_1,\rho_2\in \Ki$ such that for all $x\in \T(t_0,x_0,u)$, with $t_0\ge 0$, $x_0\in \R^n$, and $u \in \U$ such that $\|u_{t_0}\|_{\rho_1,\rho_2}<\infty$, $x$ is forward complete and $x(t)\to 0$ as $t\to \infty$.


\bibliography{/home/hhaimo/latex/strings.bib,/home/hhaimo/latex/complete_v2.bib,/home/hhaimo/latex/Publications/hernan_v2.bib}

%
\end{document}